\pdfoutput=1
\documentclass{article}

\usepackage[accepted]{template}

\usepackage[utf8]{inputenc}
\usepackage[T1]{fontenc}
\usepackage[hidelinks]{hyperref}
\usepackage{url}
\usepackage{booktabs}
\usepackage{amsfonts}
\usepackage{nicefrac}
\usepackage{microtype}
\usepackage{xcolor}
\usepackage{orcidlink}

\usepackage{amsmath}
\usepackage{graphicx}

\usepackage{multirow}%
\usepackage{amssymb}%
\usepackage{amsthm}%
\usepackage{mathrsfs}%
\usepackage{textcomp}%
\usepackage{manyfoot}%
\usepackage{booktabs}%

\usepackage{listings}%

\usepackage{amsmath, amsthm, amssymb}

\newtheorem{theorem}{Theorem}
\newtheorem{assumption}{Assumption}
\newtheorem{definition}{Definition}

\templatetitlerunning{Correcting Structural and Parametric Biases in RL Simulations of Economic Models}

\begin{document}

\twocolumn[
  \templatetitle{From Individual Learning to Market Equilibrium: Correcting Structural and Parametric Biases in RL Simulations of Economic Models}

  \templatesetsymbol{equal}{*}

  \begin{templateauthorlist}
    \templateauthor{Ruxin Chen}{nagoya}\orcidlink{0009-0007-6028-6011}\hspace{1.5em}
    \templateauthor{Zeqiang Zhang}{ulm}\orcidlink{0000-0001-8764-1123}
  \end{templateauthorlist}

  \templateaffiliation{nagoya}{Department of Economics, Nagoya University, Nagoya, Japan}
  \templateaffiliation{ulm}{Institute of Neural Information Processing, Ulm University, Ulm, Germany}

  \templatecorrespondingauthor{Zeqiang Zhang}{zeqiang.zhang@uni-ulm.de}

  \vskip 0.3in
]

\printAffiliationsAndNotice{}

\begin{abstract}
Applying Reinforcement Learning (RL) to economic modeling reveals a fundamental tension: standard RL agents do not behave as competitive price-takers. While macroeconomic models typically assume agents take aggregate conditions as given, a naive single-agent RL algorithm learns to internalize its impact on the environment, effectively exerting market power. We illustrate this divergence in a search-and-matching model with concave production, showing that a standard RL agent converges to a monopsonistic allocation rather than the competitive equilibrium. Furthermore, we show that RL’s approach to time discounting distorts standard intertemporal economic trade-offs. To resolve these issues, we propose a Calibrated Mean-Field Reinforcement Learning framework that restricts the agent to optimize against an exogenous aggregate state to restore price-taking behavior, and recalibrates the RL reward structure to reflect true economic opportunity costs.
By iteratively updating the aggregate state, our algorithm converges to a stationary fixed point that recovers the competitive equilibrium. This approach provides a method for embedding learning algorithms into dynamic general equilibrium models.
\end{abstract}

\section*{Statement}
This version of the article has been accepted for publication in \textit{Computational Economics}, after peer review but is not the Version of Record and does not reflect post-acceptance improvements, or any corrections. The Version of Record is available online at: \url{https://doi.org/10.1007/s10614-026-11464-y}. Part of this work was supported by JSPS KAKENHI Grant Number 25H00542.
\section{Introduction}

Reinforcement Learning (RL) has emerged as a powerful tool for simulating dynamic decision-making in complex environments, offering new possibilities for computational economic modeling \citep{sutton2018reinforcement, mosavi2020comprehensive}. By embedding learning agents within economic systems, RL promises to capture adaptive behavior in settings where agents interact over time under uncertainty. However, owing to conceptual differences between the two frameworks, the direct application of standard RL algorithms to economic models risks producing systematically biased results, thus failing to replicate foundational equilibrium properties. In this paper, we identify and formalize two distinct sources of bias that arise when translating economic models into RL settings: a structural bias rooted in the agent's misperception of its environment, and a parametric bias stemming from misaligned interpretations of costs and discounting.

The structural bias reflects what we term an agency-structure dilemma.
Canonical macroeconomic models, particularly in search-and-matching theory, rely on the assumption of agents who take aggregate market conditions as given \citep{pissarides2000equilibrium, krause2014modeling}. In contrast, standard single-agent RL environments, as closed-loop systems, inherently encourage the agent to exploit its environment \citep{mohiuddin2024closed}. When naively placed in a market simulation, an RL agent perceives its actions as influencing aggregate outcomes, effectively learning to behave as a strategic ``market manipulator" (monopsonist) rather than a competitive firm. This breaks the economic model's core equilibrium condition and leads to inefficient outcomes, such as under-hiring to depress wages.

In parallel, the parametric bias is subtler but equally consequential. Economic models typically incorporate both time preferences and opportunity costs of capital through parameters such as the interest rate $r$ and job destruction rate $\lambda$, which jointly determine the effective cost of job creation \citep{pissarides2000equilibrium}. Conversely, RL algorithms represent time preferences via a fixed discount factor {{$\gamma$}} and often interpret costs as per-period penalties \citep{10.1609/aaai.v33i01.33017949}. This discrepancy means that the same cost parameter $c$ has fundamentally different meanings across the two domains. We demonstrate that this translation error leads standard RL agents to miscalculate the present value of a job match, resulting in a policy that deviates from the optimal intertemporal condition.

These biases are not technical artifacts but reflect a deeper mismatch between economic theory and RL practice, which raises a central question: how can we construct a computational model that corrects both the agent's perception of its market power and its valuation of intertemporal costs?

To address these biases, a natural intuition might be to employ Multi-Agent Reinforcement Learning (MARL) to simulate a large population of firms directly \citep{ curry2022analyzing}. However, MARL suffers from the ``curse of dimensionality" and non-stationarity, making it computationally prohibitive and unstable for finding precise macroeconomic equilibria \citep{canese2021multi}. These issues motivate the search for a framework that retains the learning capabilities of RL while enforcing the consistency of economic equilibrium.

To bridge this gap, we propose a unified \textit{Calibrated Mean-Field Reinforcement Learning} (Calibrated MF-RL) framework. We reframe the economic system as a Mean-Field Game (MFG) \citep{gomes2014mean}, where a representative agent interacts with a fixed macroscopic field rather than individual competitors. This structure forces the agent to behave as a price-taker, correcting the structural bias. Simultaneously, we derive an ``effective cost" parameter ($c_{\text{eff}}$) that recalibrates the RL reward function to align with the economic definition of opportunity cost, correcting the parametric bias. Together, these corrections align the agent's learning objective with the theoretical equilibrium.

To ground our analysis, we study these simulation biases within a canonical dynamic search-and-matching framework \citep{SMITH1999456}, where a unit mass of risk-neutral workers and homogeneous firms face labor market frictions. Firms dynamically choose vacancy postings at a per-period cost $c$, and the probability of filling these vacancies depends on the aggregate market tightness $\theta$. While standard economic theory assumes firms act as price-takers regarding this aggregate field, our paper diagnoses why a standard RL agent inevitably learns to manipulate it.

In summary, this paper makes three distinct contributions to the literature on computational economics and reinforcement learning: (1) We formally identify the structural and parametric mechanisms that cause naive RL simulations to diverge from economic theory. (2) We propose the Calibrated MF-RL framework, which combines mean-field interactions with a calibrated reward structure to ensure theoretical consistency. (3) We demonstrate quantitatively that our calibrated framework successfully converges to the theoretical search-and-matching equilibrium, whereas naive RL implementations fail to do so. We further conduct ablation studies to show that correcting only one of the two biases is insufficient for equilibrium alignment, highlighting the necessity of addressing both simultaneously.

The rest of this paper is organized as follows. Section~\ref{sec:literature} positions our study within the related literature. Section~\ref{sec:model} formally establishes the theoretical baseline of the search-and-matching model. Section~\ref{sec:pitfall} provides a brief reinforcement learning setup and diagnoses the dual simulation biases via a naive RL implementation. Section~\ref{sec:framework} introduces our Calibrated Mean-Field RL framework. Section~\ref{sec:result} presents the simulation outcomes and ablation studies, and Section~\ref{sec:conclusion} concludes.

\section{Related Work}
\label{sec:literature}

The integration of RL into economic modeling has attracted growing interest in recent years, driven by the desire to simulate complex environments where agents learn and adapt over time. A number of studies have explored the use of RL to replicate or extend classical economic models, including dynamic pricing, household consumption, and labor market participation \citep{10.1145/3677052.3698621,curry2022analyzing,atashbar2023ai}. However, these applications often assume that RL can be directly applied to economic environments with minimal conceptual translation.
Our work contributes to a line of research that challenges this assumption by highlighting the theoretical and quantitative consequences of misalignments between standard RL frameworks and economic modeling principles, specifically focusing on structural and parametric dimensions.

On the structural side, several papers have adopted MARL or  MF-RL to study strategic interactions among economic agents. For example, \citet{yang2018mean} formalize MF-RL as a scalable solution for approximating Nash equilibria in large-agent systems. In parallel, recent works in computational economics \citep{angiuli2021reinforcement} have leveraged MF-RL to approximate rational expectations or decentralized coordination, particularly when analytical solutions are intractable. While these approaches emphasize scalability and equilibrium approximation, they often abstract away from the interpretation of RL rewards and cost structures in economically meaningful terms. This study builds on this literature by showing that equilibrium replication requires not only structural realism but also parametric calibration.

On the parametric side, there has been less attention paid to the economic interpretation of per-period rewards and discounting in RL. Traditional economic theory typically grounds costs and benefits in explicit opportunity cost terms, incorporating capital markets and risk-adjusted returns, as seen in the classic search-and-matching models of \citet{pissarides2000equilibrium}. In contrast, RL formulations typically collapse all future valuation into a single discount factor and apply fixed per-period costs. This difference actually has significant implications for equilibrium behavior.

Finally, our work relates to the broader literature on agent-based macroeconomic modeling (ABM), where learning agents are used to simulate aggregate dynamics \citep{tesfatsion2006agent,dwarakanath2024empirical}.  Unlike traditional ABM approaches, which often rely on rule-based heuristics, we focus on optimizing agents whose objectives are grounded in microeconomic theory but operationalized through RL.
This perspective introduces the difficulty of aligning standard RL optimization mechanisms with strict theoretical equilibrium conditions.

While a growing number of studies apply learning algorithms to ABM, the majority of these efforts rely on explicitly simulated multi-agent systems. In such models, each firm or household is represented as an individual agent, and learning is applied either through RL, evolutionary dynamics, or rule updating. Examples include works like \citet{zheng2022ai} and \citet{10975741}, which employ multi-agent deep RL to study coordination, business cycles, or labor market dynamics. Although these frameworks offer high flexibility, they often suffer from scalability issues and convergence instability, particularly when learning is decentralized and strategic interactions are dense. In contrast, our approach leverages the MF-RL framework to approximate the many-agent economy using a single representative learning agent interacting with a macroeconomic field. This allows us to capture equilibrium behavior while maintaining computational tractability and theoretical clarity.

\section{The Search-and-Matching Model}
\label{sec:model}

Following the canonical framework of \citet{SMITH1999456}, this section formalizes the dynamic search-and-matching baseline with concave production that serves as our analytical benchmark. Consider an economy populated by a unit mass of risk-neutral, infinitely-lived workers and a continuum of identical, homogeneous firms, both discounting future payoffs at a constant rate $r$.
Production requires only labor: each firm produces a single good using a strictly increasing and strictly concave production function $f(l)$, where $l$ denotes the number of workers employed. The concavity of $f$ implies diminishing marginal returns to labor, which plays a key role in wage determination and vacancy posting behavior.

Labor market frictions are modeled via a matching function $M(U, V)$ with constant returns to scale, where $U$ is the number of unemployed workers and $V$ is the total number of vacancies. The probability that a firm fills a vacancy is $q(\theta)$, and the probability that an unemployed worker finds a job is $\theta q(\theta)$, where $\theta = V / U$ denotes market tightness. In each period, existing job matches dissolve exogenously with probability $\lambda$.

Each firm chooses the number of vacancies $v_t$ in period $t$ to maximize the present value of profits. The firm's dynamic optimization problem is described by the Bellman equation,  { where $\varphi(l_t)$ denotes the value function given the current employment level $l_t$}:
\begin{equation}
\resizebox{0.89\linewidth}{!}{%
  $\begin{aligned}
    \varphi(l_t) &= \max_{v_t} \left\{ f(l_t) - w(l_t) l_t - c v_t + \frac{1}{1+r} \varphi(l_{t+1}) \right\} \\
    \text{s.t.} & \quad l_{t+1} = (1 - \lambda) l_t + q(\theta_t) v_t
  \end{aligned}$%
}
\end{equation}
where $w(l)$ is the endogenous wage function, and $c$ is the cost per vacancy.

For analytical and computational convenience, we use the following functional forms:
\begin{align}
    f(l) &= A l^\alpha, \\
    q(\theta) &= a \theta^{-\phi}, \\
    w(l) &= \frac{\eta \alpha A l^{\alpha-1}}{\eta \alpha + 1 - \eta} + (1 - \eta) b + \eta c \theta,
\end{align}
where $A$ denotes productivity, $a$ and $\phi$ denote matching efficiency and elasticity, and $\eta \in (0,1)$ governs the bargaining power of workers. We assume the unemployment benefit is endogenous to wages, characterized by a constant replacement rate $\rho$, such that $b = \rho w$.

\section{The Limitations of Naive  RL-based Simulation}
\label{sec:pitfall}

 From an RL perspective, the firm's sequential decision-making process can be modeled as a discrete-time Markov Decision Process (MDP) characterized by the standard tuple $(\mathcal{S}, \mathcal{A}, \mathcal{P}, \mathcal{R}, \gamma)$, where $\mathcal{S}$ is the state space, $\mathcal{A}$ the action space, $P$ the transition kernel, $R$ the reward function, and $\gamma \in (0,1)$ the discount factor. Unlike traditional dynamic programming, which requires full knowledge of transition probabilities, a reinforcement learning agent aims to learn an optimal policy $\pi^*$ to maximize the expected cumulative discounted reward $\mathbb{E}_\pi [\sum_{t=0}^{\infty} \gamma^t R(s_t, a_t)]$ purely through continuous environmental interaction. However, standard single-agent RL formulations inherently assume a fixed and stationary environment. In dynamic economic contexts where firms collectively shape aggregate labor market conditions, this assumption fails due to endogeneity.

To illustrate the theoretical distortions caused by this mismatch, this section examines what happens when the search-and-matching model from Section~\ref{sec:model} is naively translated into a standard single-agent RL simulation. We begin by computing the unique competitive steady-state equilibrium of the analytical model as our benchmark, and then contrast it against the outcome of the RL training under identical primitives. Surprisingly, the standard RL agent converges to a markedly different policy that extracts excess profits. To diagnose this divergence, we conduct a mathematical decomposition of the agent's objective function and identify two distinct sources of error: a \textit{structural bias} caused by endogeneity of market aggregates, and a \textit{parametric bias}  due to misaligned cost and discounting assumptions. These mechanisms constitute what we define as the dual simulation bias problem.

\subsection{The Theoretical Benchmark}

The job-creation condition in the economic model is derived from the firm's Bellman equation and envelope condition:
\begin{equation}
    f'(l) - w(l) - w'(l)l = \frac{(r + \lambda)c}{q(\theta)}.
\end{equation}

This condition characterizes the firm's optimality under the assumption that $\theta$ is an exogenous market parameter. Combined with equilibrium conditions for labor market flows and wage determination, the model yields a unique steady-state characterized by $(l^*, u^*, q^*, w^*, v^*, \theta^*)$.
	\begin{table}[!htbp]
		\caption{Default experimental parameters.}
		\label{tab:1}
		\begin{center}
		\begin{tabular}{c|c|c}
		\toprule
		\textbf{Symbol} & \textbf{Description}& \textbf{Value} \\
		\midrule
            $A$ & Productivity & $1$\\
		$a$ & Matching efficiency & $0.471$\\
		$\alpha$ & Parameter in production function & $0.667$\\
		$\lambda$ & Separation rate & $0.0144$\\
		$\eta$ & Worker's bargaining power & $0.6$\\
		$c$ & Vacancy cost & {$0.268$}\\
		$\phi$ & Matching elasticity& $0.6$\\
            $   r$ & Interest rate& $0.01$\\
           { $   \rho$} & {Unemployment replacement rate}& {$0.6$}\\

		\bottomrule

		\end{tabular}
		\end{center}
		\end{table}

{  Our calibration strategy closely follows \citet{kudoh2019employment} and \citet{economies13100302}.} Using the calibrated parameters in Table~\ref{tab:1}, we solve the full system of equations and obtain the theoretical steady state reported in Table~\ref{tab:2} { (see Appendix~\ref{appendix:steady_state} for calibration details and steady-state values)}. Notably, the equilibrium value of market tightness is {$\theta^* = 0.78$}.

\subsection{A Naive RL Simulation and its Divergence}
\label{sec:naive_translation}

To test whether RL can reproduce this equilibrium, we construct a single-agent RL environment that mimics the firm's optimization problem. The RL agent observes the current labor stock $l_t$ and chooses a vacancy posting level $v_t$, transitioning to $l_{t+1} = (1 - \lambda) l_t + q(\theta_t) v_t$, where $\theta_t = v_t / u_t$ and $u_t = 1 - l_t$. The reward is defined as:
{
\begin{equation}
    R_t = f(l_t) - w(l_t, \theta_t)l_t - c v_t,
\end{equation}}

and the objective is to maximize the expected cumulative reward:
{
\begin{equation}
    \mathbb{E} \left[\sum_{t=0}^\infty \gamma^t R_t \right].
\end{equation}}

The experiment details are summarized in Appendix~\ref{appendix:RL}. The learned steady-state of the RL agent can then be compared to the theoretical benchmark. We also calculate the theoretical steady-state reward, which corresponds to the firm's flow profit in the competitive equilibrium. Using the equilibrium values derived in Table~\ref{tab:2}, the theoretical per-period reward is:
\begin{equation}
    R^* = f(l^*) - w(l^*, \theta^*)l^* - c v^* \approx 0.168.
\end{equation}

\begin{figure}[htbp]
\centering
\includegraphics[width=\columnwidth]{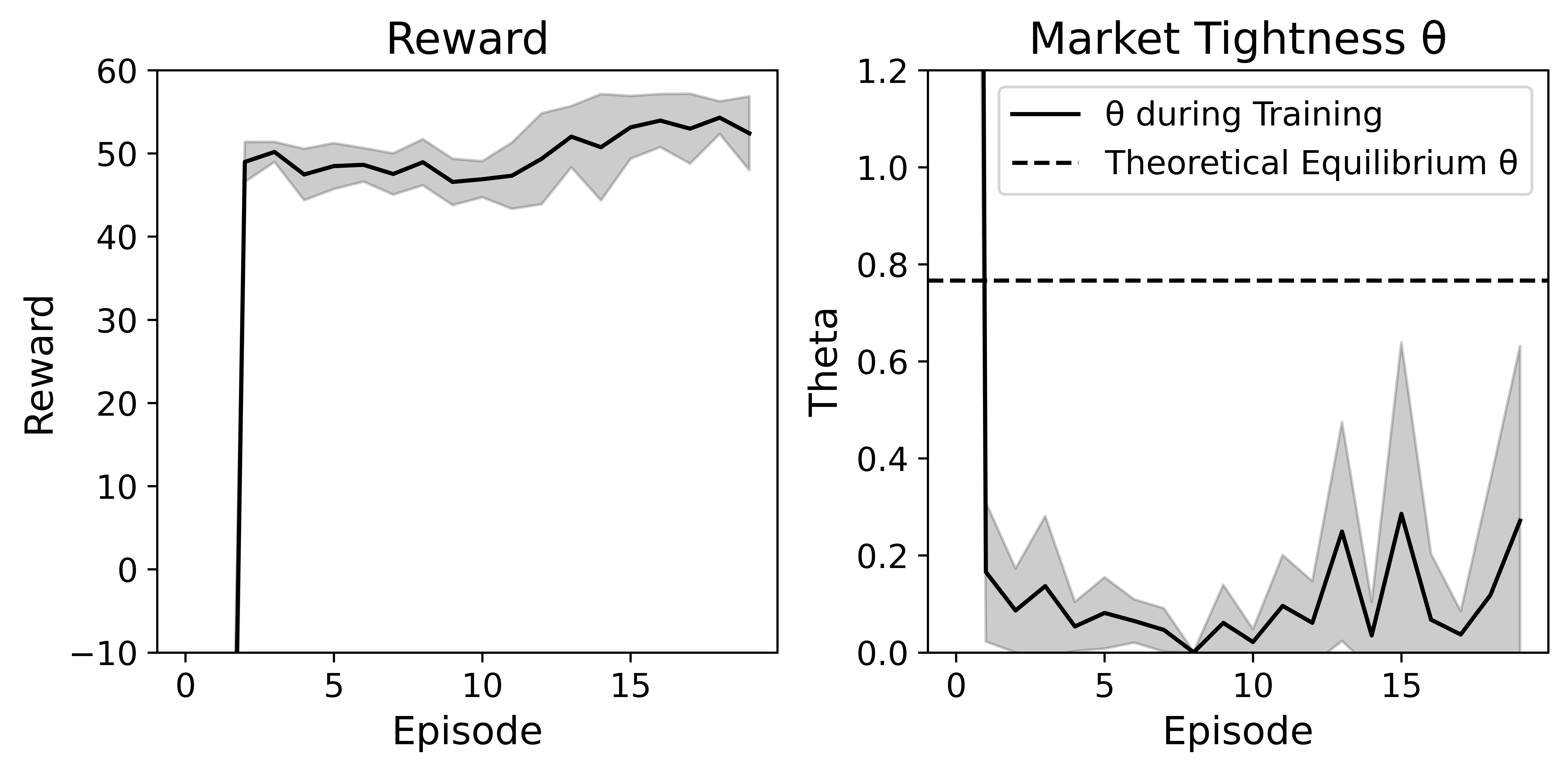}
\caption{Market tightness $\theta$: theoretical benchmark vs. RL outcome. The left panel illustrates the variation of reward during the training process, where the solid line represents the reward for the final step, the dotted line for the accumulated reward, and a dashed horizontal line for the theoretical reward value $0.168$ of the steady-state derived from economic models. The stabilization of reward indicates that the agent has nearly converged to the optimal policy. The right panel depicts the changes in $\theta$ over the course of training; once the agent's strategy has converged, the market tightness fluctuates around $0.1$, which is significantly lower than the theoretical value of $0.78$ derived from economic models. The shaded regions indicate the standard deviations from five independent runs.}
\label{fig:compare_theta}
\end{figure}

The steady-state benchmarks of the reward $R^*$ and market tightness $\theta^*$ are depicted as a dashed horizontal line in Figure~\ref{fig:compare_theta}. The RL agent converges to a reward level significantly higher than this benchmark, and a market tightness far below the equilibrium level. This indicates that the agent has not failed to optimize; in contrast, it has converged to a policy that extracts supernormal profits. This observation compels us to reframe the problem: what optimization problem is the agent actually solving?

\paragraph{Structural Bias: The ``Market Manipulator" Effect}

To understand this discrepancy, we analyze the RL agent's objective in steady state:
{
\begin{equation}
    J(l) = \frac{1}{1 - \gamma} \left[f(l) - w(l, \theta)l - c v \right].
\end{equation}
}

Differentiating with respect to $v$, and applying the chain rule:

\begin{equation}
\resizebox{0.89\linewidth}{!}{$
\label{eq:rl_foc}
\frac{\partial J(l)}{\partial v} = \frac{1}{1-\gamma} \left[ (f' - w - w' l) \cdot \frac{\partial l}{\partial v} - \frac{\partial \theta}{\partial v} \cdot l \cdot \frac{\partial w}{\partial \theta} - c \right] = 0.$}
\end{equation}
Since $l = \frac{q(\theta)}{\lambda} v$, we compute:
\begin{equation}
\frac{\partial l}{\partial v} = \frac{q}{\lambda} + \frac{v}{\lambda} \cdot \frac{dq}{d\theta} \cdot \frac{\partial \theta}{\partial v}.
\end{equation}

Substituting into (\ref{eq:rl_foc}), we obtain:
\begin{equation}
    (f' - w - w'l) = \frac{c + \frac{\partial \theta}{\partial v} \cdot l \cdot \frac{\partial w}{\partial \theta}}{\frac{q}{\lambda} + \frac{v}{\lambda} \cdot \frac{dq}{d\theta} \cdot \frac{\partial \theta}{\partial v}}.
\end{equation}

If the agent internalizes that $\theta$ depends on $v$, i.e., $\frac{\partial \theta}{\partial v} \ne 0$, it will deliberately reduce $v$ to suppress $\theta$ and lower the wage $w$. In this case, the RL agent converges to a monopsonistic policy, which departs from the competitive equilibrium.

Thus, we identify a structural bias: the RL agent acts as a “market manipulator” because the simulation environment allows it to control $\theta$, violating the “tightness-taker” assumption central to the economic model.

\paragraph{Parametric Bias: The Cost-Discounting Mismatch}

A second source of error lies in the treatment of cost and time discounting. In the economic model, the expected discounted cost of hiring via a vacancy is:
\begin{equation}
    \frac{(r + \lambda)c}{q(\theta)},
\end{equation}
where $(r + \lambda)$ reflects the opportunity cost of capital and the job separation hazard.

However, in the RL simulation, the cost $c$ is discounted using the factor {$\gamma$} and summed over time. Since expected job duration is $1/\lambda$, the effective cost becomes:
\begin{equation}
    \frac{\lambda c}{q(\theta)}.
\end{equation}

This mismatch in cost specification leads the RL agent to underestimate the true economic cost of vacancy posting. {
We can quantify this parametric bias as a scaling factor $\chi$, representing the ratio between the theoretical economic cost and the RL agent's perceived cost:

\begin{equation}
    \label{eq:param_bias_factor}
    \chi = \frac{(r + \lambda)c / q(\theta)}{\lambda c / q(\theta)} = 1 + \frac{r}{\lambda}.
\end{equation}

This equation (\ref{eq:param_bias_factor}) suggests that the RL environment fails to faithfully represent the economic objective because it treats costs as immediate rather than capitalizing them over the job duration, effectively scaling the cost down by a factor of $1/\chi$.}

In combination, these two biases explain the divergence between the RL simulation and the analytical equilibrium. In the next section, we propose a unified correction framework based on mean-field games and cost calibration to eliminate both sources of error.

\section{Calibrated Mean-Field Reinforcement Learning}
\label{sec:framework}
Having diagnosed the structural and parametric biases that arise from a naive translation of economic theory into RL simulations, we now propose a unified correction framework. Our approach addresses:

\begin{itemize}
    \item the \textbf{structural bias} by embedding the agent within a \emph{MFG} formulation, thereby ensuring that aggregate variables like $\theta$ are treated as exogenous during optimization but endogenously updated across iterations;
    \item the \textbf{parametric bias} by introducing an \emph{effective cost parameter} $c_{\text{eff}}$ that reflects the economically consistent long-run cost of hiring under the correct intertemporal assumptions.
\end{itemize}

Jointly, these corrections restore alignment between the RL simulation and the theoretical model, enabling the agent to learn behavior consistent with the competitive equilibrium.

\subsection{Correcting the Structural Bias: A Mean-Field Formulation}

MFG theory provides a scalable solution concept for systems with a large number of interacting agents. Developed independently by \citet{lasry2007mean} and \citet{huang2006large}, MFGs model each agent as optimizing in response to a time-varying mean field that summarizes the behavior of the population. In equilibrium, the mean field is consistent with the distribution of agents generated by their optimal policies, forming a fixed point.

Let $\mu_t$ denote the mean field at time $t$—for example, the distribution of states or actions. The representative agent solves:
\begin{equation}
    \pi^* = \arg\max_{\pi} \; \mathbb{E}_\pi \left[ \sum_{t=0}^\infty \gamma^t R(s_t, a_t, \mu_t) \right]
\end{equation}
subject to transition dynamics $s_{t+1} \sim P(s_{t+1} | s_t, a_t, \mu_t)$ and a consistency condition:
\begin{equation}
    \mu_{t+1} = \mathcal{F}(\mu_t, \pi^*).
\end{equation}

In practice, solving MFGs analytically can be intractable. \textit{Mean Field Reinforcement Learning} (MF-RL) offers a practical alternative by approximating the MFG fixed point via an iterative procedure. { \citet{a15030073} provide a comprehensive framework for this approach, demonstrating how RL algorithms can be effectively adapted to solve Mean-Field Games by decoupling the agent's optimization from the population dynamics. Following their iterative methodology, the procedure is as follows}:

\begin{enumerate}
    \item Initialize a mean field $\mu^{(0)}$;
    \item Given $\mu^{(k)}$, solve the single-agent RL problem to obtain policy $\pi^{(k)}$;
    \item Update the mean field via $\mu^{(k+1)} = \mathcal{F}(\mu^{(k)}, \pi^{(k)})$;
    \item Repeat until convergence to $(\pi^*, \mu^*)$.
\end{enumerate}

This fictitious-play-style algorithm is particularly well-suited to economic environments where agents are atomistic yet strategically interdependent. In our context, MF-RL provides a principled way to reconcile RL with equilibrium behavior in search-and-matching labor markets. To prevent the agent from behaving as a ``market manipulator”, we adopt a mean-field approximation of the full multi-agent economy. In the limit of a continuum of firms, each agent is infinitesimal and thus treats aggregate variables—like market tightness $\theta$—as exogenous when solving its optimization problem.

Formally, denote the mean field at iteration $k$ as $\theta^{(k)}$. The RL agent then solves:
\begin{equation}
\resizebox{0.89\linewidth}{!}{%
  $\begin{aligned}
    \pi^{(k)} & = \arg\max_{\pi} \mathbb{E}_{\pi}\left[ \sum_{t=0}^\infty \gamma^t \left(f(l_t) - w(l_t, \theta^{(k)}) l_t - c v_t \right) \right], \\
    \text{s.t.} & \quad l_{t+1} = (1 - \lambda) l_t + q(\theta^{(k)}) v_t.
  \end{aligned}$%
}
\end{equation}

After solving for the optimal policy $\pi^{(k)}$, we simulate a population of agents using this policy to compute the implied new aggregate $\theta^{(k+1)}$. The fixed point $\theta^*$ is reached when $\theta^{(k+1)} \approx \theta^{(k)}$.

This framework ensures that the RL agent does not internalize the effect of its own actions on macro variables like $\theta$, thus preserving the atomistic assumption of the economic model.

\subsection{Correcting the Parametric Bias: Cost Calibration via $c_{\text{eff}}$}
\label{sec:param_bias}
One of the key conceptual mismatches between standard economic models and RL lies in how cost and intertemporal trade-offs are encoded. In economic models, firms evaluate investment decisions not just by future returns, but by comparing them against the opportunity cost of capital. In contrast, standard RL implementations typically internalize time via a discount factor $\gamma$, without explicitly accounting for alternative uses of resources—such as financial savings—available to economic agents.

To illustrate this divergence, consider the firm's decision to create a vacancy. In the economic model, each vacancy costs $c$ per period while it remains unfilled or occupied. The firm expects to recover this cost over the duration of a job, which terminates at a Poisson rate $\lambda$. However, the firm must also consider the alternative: investing the same funds elsewhere, earning a return at rate $r$. The relevant cost for the firm is therefore not just the flow cost $c$, but the \emph{present value of foregone investment income} over the expected lifetime of the job.

This reasoning yields the job creation condition:
\begin{equation}
f'(l) - w(l) - w'(l)l = \frac{(r + \lambda)c}{q(\theta)},
\end{equation}
where $(r + \lambda)$ reflects both the hazard of separation and the interest rate—together determining the effective decay rate of the employment relationship's present value. This structure arises from the Bellman equation of the firm that embeds both job turnover and capital cost.

In contrast, in the naive RL implementation, the firm-agent maximizes
\begin{equation}
J = \mathbb{E} \left[ \sum_{t=0}^\infty \gamma^t \left(f(l_t) - w(l_t, \theta_t)l_t - c v_t \right) \right],
\end{equation}
where $\gamma \approx \frac{1}{1 + r}$ serves as a discount factor. However, the cost term $c$ is applied per period without adjustment. While the job separation rate $\lambda$ is often included in the environment's transition dynamics, it affects only the firm's employment flow, not how the cost is valued.

This creates a subtle but important asymmetry: the RL agent understands that jobs are finite-lived (via $\lambda$), but fails to account for the fact that vacancy costs represent foregone capital returns (via $r$). As a result, it undervalues the true economic cost of creating a job. The agent behaves as if the entire budget is held in “real terms”, ignoring the opportunity to allocate capital elsewhere.

To reconcile this difference, we introduce an \emph{effective cost} parameter $c_{\text{eff}}$, which internalizes both the job’s expected duration and the capital opportunity cost. Specifically, we define:
\begin{equation}
c_{\text{eff}} := \left(1 + \frac{r}{\lambda} \right) c,
\end{equation}
which implies that the true economic cost of a vacancy is scaled up by the firm’s inability to invest that capital elsewhere. The term $\frac{r}{\lambda}$ reflects how the firm trades off posting a vacancy versus saving the money at interest rate $r$ over the job’s expected life $\frac{1}{\lambda}$.

From an economic standpoint, $c_{\text{eff}}$ represents the \textit{capitalized long-run cost} of a job match. It answers the question: “what is the present value of committing $c$ per period over a job duration of $\frac{1}{\lambda}$, while giving up a return of $r$?” This calibration ensures that the RL agent evaluates its hiring decision in a way that is equivalent to the firm in the economic model.

Accordingly, we modify the reward function in the RL environment as:
\begin{equation}
r_t = f(l_t) - w(l_t, \theta_t)l_t - c_{\text{eff}} v_t.
\end{equation}

This change is more than a numerical fix; it embeds an economic worldview into the agent's objective. It ensures that the simulated firm's hiring policy reflects not only the timing of profits and job turnover, but also the broader capital allocation logic inherent to investment decisions in equilibrium macroeconomic models.

This adjustment is critical to bridge the second of the two core modeling biases we identify in this paper. Without this calibration, RL-based simulations will consistently produce hiring policies that deviate from their economic counterparts—not because of poor optimization, but because the agent is solving a different, and economically inconsistent, objective.

\subsection{The Combined Algorithm: Calibrated MF-RL}

We now combine the structural and parametric corrections into a unified procedure, which we call \emph{Calibrated Mean-Field Reinforcement Learning} (Calibrated MF-RL). The agent solves its problem under a fixed mean field $\theta^{(k)}$ using a calibrated cost parameter $c_{\text{eff}}$, and the mean field is iteratively updated until convergence.

\begin{algorithm}[tb]
  \caption{Calibrated Mean-Field RL}
  \label{alg:cal_mf_rl}
  \begin{algorithmic}
    \STATE Compute $c_{\text{eff}} = \frac{(r + \lambda)}{\lambda} c$
    \STATE Initialize mean field $\theta^{(0)}$
    \FOR{$k = 0, 1, 2, \dots$ {\bfseries until} convergence}
    \STATE Solve RL problem using $r_t = f(l_t) - w(l_t, \theta^{(k)}) l_t - c_{\text{eff}} v_t$
    \STATE Obtain optimal policy $\pi^{(k)}$
    \STATE Simulate firms under $\pi^{(k)}$ to compute updated $\theta^{(k+1)}$
    \ENDFOR
  \end{algorithmic}
\end{algorithm}

This algorithm guarantees consistency between the RL objective and the economic theory in both dimensions:
\begin{itemize}
    \item \textbf{Structure:} By fixing $\theta$ during RL training, the agent behaves atomistically.
    \item \textbf{Parameters:} By calibrating $c_{\text{eff}}$, the reward aligns with intertemporal optimization.
\end{itemize}

In the next section, we empirically validate the effectiveness of this framework, showing that the resulting simulation converges to the theoretical steady-state equilibrium.

\subsection{Convergence of the RL-based Mean Field Solver}

Having introduced a corrected RL framework that aligns agent incentives with economic equilibrium, it is natural to ask: under what conditions does this procedure converge to a valid equilibrium? In this subsection, we provide a theoretical justification for the convergence of our RL-based mean-field learning algorithm using a fixed-point argument.

Recall that in our environment, the market tightness $\theta$ serves as the key aggregate variable mediating strategic interaction. At each iteration, we fix a value of $\theta$, solve a Markov Decision Process (MDP) to obtain the agent’s optimal policy $\pi^*_\theta$, and then update the macro variable based on population-level behavior under $\pi^*_\theta$. Formally, this defines a composite mapping:
\begin{equation}
    \Psi(\theta) = \Phi(\pi^*_\theta),
\end{equation}

where $\Phi$ is the environment-level aggregator that computes the updated mean field from the agent's behavior. A fixed point $\theta^*$ of $\Psi$ corresponds to a \textit{Mean Field Equilibrium} in which individual behavior and macro dynamics are mutually consistent.

To analyze the convergence of this iterative process, we invoke Theorem~\ref{thm:convergence} (see Appendix~\ref{app:convergence}). We assume:
\begin{itemize}
    \item[(A1)] The policy map $\theta \mapsto \pi^*_\theta$ is Lipschitz continuous. This assumption is justified in our case by the smoothness of the reward function and the use of stable, regularized RL algorithms.
    \item[(A2)] The environment response map $\pi \mapsto \Phi(\pi)$ is Lipschitz. In our labor market setting, this is ensured by the matching function $M(U, V)$ being smooth and strictly concave, and the macro variables (e.g., aggregate employment and unemployment) being continuous aggregations over a population of homogeneous firms.
\end{itemize}

{
Under these conditions, the composite mapping $\Psi$ is a contraction whenever the product of Lipschitz constants $L_1 L_2 < 1$. While rigorously verifying these global Lipschitz bounds for deep neural networks is analytically intractable, the stability of our implementation is structurally supported by the soft update (Polyak averaging) mechanism inherent to DDPG. This mechanism ensures that the agent's policy evolves gradually rather than abruptly, effectively enforcing a form of local continuity consistent with Assumption (A1). Furthermore, the smoothness of the underlying strictly concave production and matching functions ensures that the environment response is well-behaved (Assumption A2). Consequently, while we may not guarantee a precise mathematical fixed point due to approximation errors, we expect and empirically observe approximate convergence to a stable neighborhood of the theoretical equilibrium.
}

{
By the Banach Fixed Point Theorem, the mean field sequence $\theta_{k+1} = \Psi(\theta_k)$ converges exponentially to a fixed point $\theta^*$. In the specific context of our search-and-matching model, where the theoretical equilibrium is known to be unique, this fixed point $\theta^*$ corresponds globally to the equilibrium of the corrected system.

More generally, in economic environments characterized by multiple equilibria, this result implies that our algorithm functions as a local solver, akin to standard numerical methods used in DSGE modeling (e.g., Newton-type solvers) where the converged solution may depend on initialization. }

\section{Simulation Results}
\label{sec:result}
In this section, we evaluate the effectiveness of our proposed Mean-Field RL framework, combined with cost calibration, in resolving both structural and parametric biases identified in Section~\ref{sec:naive_translation}. We present simulation results that demonstrate convergence toward the theoretical economic equilibrium.

We consider the same search-and-matching economy described in Section~\ref{sec:model}. The RL agent is trained using a standard Deep Deterministic Policy Gradient (DDPG) algorithm with experience replay \citep{DBLP:journals/corr/LillicrapHPHETS15}. For the mean-field solver, we initialize a guess for the aggregate policy (market tightness $\theta$), solve the agent’s RL problem given that field, and update the aggregate statistics accordingly. This process is repeated iteratively until convergence. The effective vacancy cost $c_{\text{eff}}$ is computed as described in Section~\ref{sec:param_bias}. We summarize the details of the simulation in Appendix~\ref{appendix:RL}.

Figure~\ref{fig:baseline_compare} shows the key steady-state variables—market tightness $\theta$, unemployment $u$, and vacancy $v$—computed from the theoretical model and those generated by the fully corrected RL simulation (using MFG and $c_{\text{eff}}$). The simulation closely aligns with the theoretical benchmark, demonstrating that our method successfully replicates the equilibrium predicted by economic theory. We also conduct ablation studies (see Appendix~\ref{appendix:ablation}) to show that correcting only one of the two biases is insufficient to recover the correct behavior.

\begin{figure}[ht]
    \centering
    \includegraphics[width=\columnwidth]{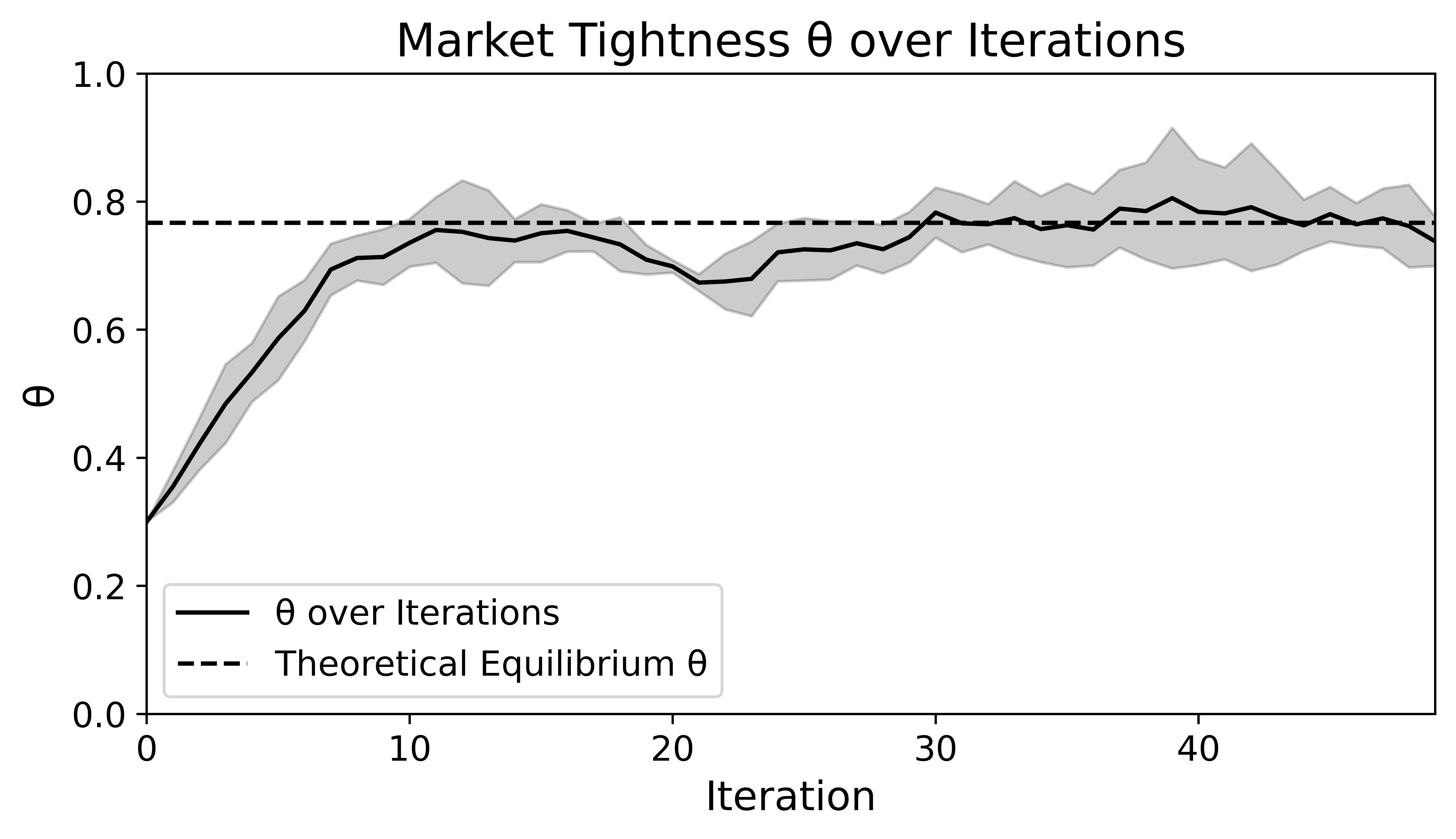}
    \caption{Comparison between theoretical equilibrium and fully corrected RL simulation. The figure depicts the changes in $\theta$ over iterations; once it has converged, the market tightness fluctuates around the theoretical value of $0.78$ derived from economic models. The shaded regions indicate the standard deviations from five independent runs.}
    \label{fig:baseline_compare}
\end{figure}

These results provide strong evidence supporting our central claims. First, we observe that a naive RL implementation fails to reproduce the theoretical economic equilibrium due to the presence of both structural and parametric modeling biases. Each of these biases produces a qualitatively distinct distortion: structural bias leads the agent to behave strategically as a market manipulator, while parametric bias causes it to misestimate the true economic cost of posting vacancies. Crucially, we find that correcting only one of these issues is insufficient. It is only when both the mean-field interaction structure and the calibrated cost formulation are jointly applied that the learned agent behavior aligns with the competitive equilibrium, effectively simulating a ``price-taking'' firm.

This reinforces our broader thesis: simulating economic models with RL demands careful alignment between the agent’s learning objective and the theoretical underpinnings of the economic environment, including both institutional assumptions and the interpretation of cost parameters.

\section{Conclusion}
\label{sec:conclusion}
This paper demonstrates that naively applying RL to economic models can produce systematically biased outcomes due to a mismatch between the agent's optimization structure and the assumptions embedded in economic theory. We identified two core biases: a \textit{structural bias}, where the agent learns to strategically manipulate market tightness; and a \textit{parametric bias}, stemming from the misalignment between economic opportunity costs and RL discounting.
{

Our estimates and analyses yield three key findings:
\begin{itemize}
    \item \textbf{Quantification of Bias:} In the naive simulation, the RL agent behaves as a monopsonist, significantly deviating from the theoretical competitive equilibrium. This confirms that standard RL agents exploit their environment if not structurally constrained.
    \item \textbf{Effectiveness of Calibration:} By implementing the proposed \textit{Calibrated Mean-Field Reinforcement Learning} framework, which enforces atomistic behavior and corrects cost perception ($c_{\text{eff}}$), the simulation successfully converges to the theoretical benchmark, reproducing the target market tightness and unemployment rates with high precision.
    \item \textbf{Necessity of Dual Correction:} Our ablation studies reveal that addressing either bias in isolation is insufficient. Structural correction alone leads to over-hiring due to cost underestimation, while parametric correction alone results in excessive market manipulation. Only the joint application of both corrections ensures equilibrium consistency.
\end{itemize}

In conclusion, this study provides a rigorous methodology for integrating learning agents into economic models. By aligning the RL objective with economic theory, we offer a tractable path for future research in computational social science, particularly for models requiring consistent micro-foundations. Future work may extend this framework to environments with multiple equilibria or heterogeneous agents.}

\bibliography{sample.bib}
\bibliographystyle{plainnat}
\newpage
\appendix

\section{Economic Model Equilibrium}
\subsection{Equations for Steady State}
\label{appendix:Equations}
{
The steady-state equilibrium is defined by a system of simultaneous equations that characterize the firm's optimal job creation condition, wage determination, and labor market flows. By setting all time-varying variables to constants and combining the Bellman equations, we obtain the following system which allows us to solve for the endogenous variables $\{l, u, v, q, \theta, w, b\}$:}
\begin{equation}
\resizebox{0.89\linewidth}{!}{%
  $\begin{aligned}
    \label{eq1}
        \alpha A l^{\alpha-1}-\frac{\eta\alpha^2Al^{\alpha-1}}{\eta \alpha+1-\eta}-(1-\eta)b-\eta c \theta -\frac{(r+\lambda)c}{q(\theta)}&=0\\
        w-\frac{\eta\alpha Al^{\alpha-1}}{\eta \alpha+1-\eta}-(1-\eta)b-\eta c \theta&=0\\
        \lambda l- qv&=0\\
        q-a\theta^{-\phi}&=0\\
        \theta -\frac{v}{u}&=0\\
        l+u-1&=0\\
        b-\rho w&=0\\
  \end{aligned}$%
}
\end{equation}

{
\subsection{Parameter Calibration}

Given that the primary objective of this study is to diagnose and correct systematic biases in RL-based economic simulations, our findings are robust to specific parameter choices, as long as they reflect realistic economic conditions. For this study, we calibrate the model to match key facts of the Japanese labor market.

We adopt standard parameter values from the literature. Following \citet{kudoh2019employment} and \citet{LIN201486}, the labor share parameter is set to $\alpha = 2/3$, and the matching elasticity to $\phi = 0.6$. To satisfy the \citet{Hosios1990on} efficiency condition, we set the worker's bargaining power equal to the matching elasticity ($\eta = 0.6$). The unemployment replacement rate is set to $\rho = 0.6$, as reported in \citet{Stephen1997Unemployment}; and the quarterly separation rate is $\lambda = 0.0144$, consistent with estimates by \citet{MIYAMOTO2011214}. We normalize productivity $A=1$ and set the quarterly interest rate $r=0.01$.

To pin down the remaining parameters, we target a steady-state market tightness of $\theta = 0.78$, as reported by \citet{MIYAMOTO2011214}. Matching the empirical quarterly job-finding rate of $0.426$ \citep{lin2012gross} implies a matching efficiency of $a = 0.471$. Finally, the vacancy posting cost is residually determined as $c = 0.268$ to satisfy the job creation condition at the targeted equilibrium. These parameters are held constant across both the theoretical benchmark and the RL experiments.
}
{
\subsection{Theoretical Steady State}
\label{appendix:steady_state}
By substituting the calibrated parameter values into the system of equations derived in Appendix~\ref{appendix:Equations}, we numerically solve for the unique competitive equilibrium. Table~\ref{tab:2} reports the resulting steady-state values for the key macroeconomic variables.

	\begin{table}[!htbp]
		\caption{Theoretical steady state.}
		\label{tab:2}
		\begin{center}
		\begin{tabular}{c|c|c|c|c|c}
		\toprule
		\textbf{$l$} & \textbf{$u$}& \textbf{$q$} &\textbf{$w$} &\textbf{$v$} &\textbf{$\theta$} \\
		\midrule
            $0.967$ & $0.033$ & $0.547$ &$0.831$&$0.025$&$0.78$\\

		\bottomrule

		\end{tabular}
		\end{center}
		\end{table}}
\section{RL Simulation settings}
\label{appendix:RL}
In our simulation, {we use a DDPG algorithm} to optimize the decision-making policy of the learning agent. This section details the architectural and training configurations used in our implementation.

\subsection{Agent Architecture}

We implement an actor-critic architecture { in Python with PyTorch \citep{paszke2019pytorch}} as follows:

\begin{itemize}
    \item \textbf{Actor Network}: The actor network maps observed states {$s_t$} to continuous actions. It consists of a 3-layer feedforward neural network with ReLU activations and a final $\tanh$ output layer to bound actions:
    {
    \begin{align*}
        \text{Actor: } \quad s_t &\rightarrow o^A_1:= \text{ReLU}(W_1 s_t + b_1)\\ &\rightarrow o^A_2:= \text{ReLU}(W_2 o^A_1 + b_2)\\ &\rightarrow a_t:=(\tanh(W_3 o^A_2 + b_3)+1)/2
    \end{align*}}
    \item \textbf{Critic Network}: The critic takes state-action pairs and estimates their Q-values. It is also a 3-layer feedforward network, where the input is the concatenation of the state and action vectors:
    {
    \begin{align*}
        \text{Critic: } \quad (s_t,a_t) &\rightarrow o^C_1:=\text{ReLU}(W_1 [s,a] + b_1)\\ &\rightarrow o^C_2:=\text{ReLU}(W_2 o^C_1 + b_2) \\&\rightarrow Q(s_t,a_t):=W_3 o^C_2 + b_3
    \end{align*}}
\end{itemize}

Here $o^A_i$ or $o^C_i$ means the output value of the $i$th layer of the actor or critic networks. The actor and critic networks each have 256 hidden units per layer. The actor is optimized using Adam with a learning rate of $5 \times 10^{-5}$, while the critic uses a learning rate of $5 \times 10^{-4}$. These values were selected based on the standard actor-critic practice where the critic is updated at a faster rate to provide accurate value estimates for policy updates. Target networks are updated using a Polyak averaging factor $\tau = 0.005$. Discount factor is set to $\gamma = 0.99$.

\subsection{Training Procedure}

We use a replay buffer with a maximum capacity of $10^5$ transitions to store the agent’s experiences. At each training step, we sample mini-batches of 256 transitions to update the networks. The agent performs soft updates of the target networks and alternates between critic and actor updates following the standard DDPG procedure.

Each training iteration (used in mean-field reinforcement learning) consists of 50 independently simulated episodes, each of length 200 steps.

\subsection{State and Action Spaces}

The environment state at each time step $t$ is defined as the vector $o_t = (l_t, u_t)$, where:
\begin{itemize}
    \item $l_t$ denotes the current employment rate,
    \item $u_t$ is the current unemployment rate.
\end{itemize}

The action space consists of a single continuous action $v_t$, which represents the vacancy posting decision made by the firm agent. The output of the actor network is interpreted as the normalized vacancy posting intensity.

{
To ensure the numerical stability and economic validity of the simulation, we explicitly address the boundedness and positivity of the state and action space from both theoretical and algorithmic perspectives.

\paragraph{Economic Model Constraints.}
The search-and-matching framework imposes natural constraints on the state variables:
\begin{itemize}
    \item \textbf{Boundedness of Employment:} The state variable $l_t$ represents the employment rate. Given the normalization of the total labor force ($l_t + u_t = 1$), the employment level is strictly bounded within the closed interval $l_t \in [0, 1]$.
    \item \textbf{Positivity of Flows:} Economic quantities such as vacancies $v_t$ and hires $H_t = q(\theta_t)v_t$ are theoretically non-negative.
\end{itemize}

\paragraph{Algorithmic Implementation in DRL.}
Since standard neural networks operate as a black box on the domain $(-\infty, +\infty)$, the mathematical machinery for stability in DDEs (e.g., Routh-Hurwitz criteria) \citep{wang2025modeling,li2025hopf} differs fundamentally from the contraction mapping principles used in our discrete Bellman equations. Therefore, we implement specific architectural choices to enforce these constraints during the learning process:
\begin{itemize}
    \item \textbf{Action Positivity:} The Actor network determines the vacancy posting $v_t$. To prevent the agent from proposing economically impossible negative vacancies, we apply a Tanh activation function to the output layer:
    \begin{equation}
        v_t:=\frac{1}{2}[\tanh(O_{NN})+1],
    \end{equation}
     where $O_{NN}$ is the output of the neural network. This guarantees $v_t \in (0,1)$ for all time steps $t$.

    \item \textbf{State Boundedness Preservation:} The environment's transition dynamics are governed by the law of motion: $l_{t+1} = (1-\lambda)l_t + q(\theta_t)v_t$. Provided that the initial state $l_0 \in [0,1]$ and the vacancy filling probability $q(\theta_t) \in [0,1]$, the non-negativity of $v_t$ ensures that the employment rate remains well-behaved and bounded within the unit interval throughout the simulation.
\end{itemize}
}

{
\section{Supplementary Analyses}
\label{appendix:supplementary}
In this appendix, we provide additional experimental results to further validate the robustness of our framework and interpret the agent's learned behavior.}

\subsection{Ablation Study}
\label{appendix:ablation}

To verify that both corrections are necessary, we conduct ablation experiments:

\begin{itemize}
    \item \textbf{Only Structural Correction (MFG only):} The agent solves a mean-field game, but the cost parameter remains uncalibrated. This leads to over-optimistic hiring, as the agent underestimates the true economic cost of vacancies. The resulting $\theta$ is too high relative to the theoretical benchmark.
    \item \textbf{Only Parametric Correction ($c_{\text{eff}}$ only):} The cost function is calibrated, but the environment is modeled as a single-agent MDP. The agent behaves as a market manipulator, suppressing vacancy creation to lower wages. The resulting $\theta$ is significantly below equilibrium.
\end{itemize}

Figure~\ref{fig:ablation} shows the learned $\theta$ in these two conditions. Compared with Figure~\ref{fig:baseline_compare}, { we conclude} that only when both corrections are applied does the simulation converge to the correct value.

\begin{figure}[ht]
    \centering
    \includegraphics[width=\columnwidth]{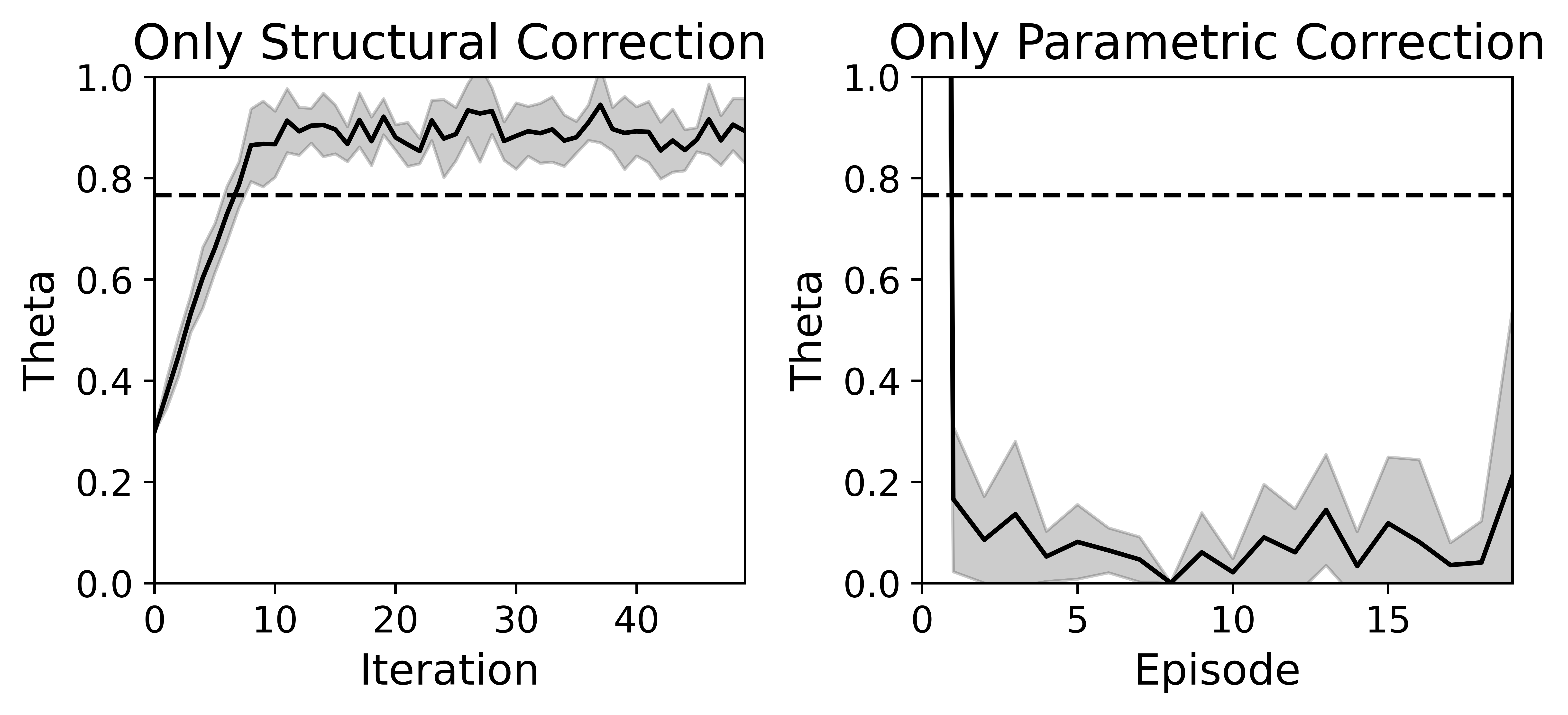}
    \caption{Market tightness $\theta$ across different simulation settings. The left panel shows the situation where there is only structural correction, while the right panel shows the situation where there is only parametric correction. The shaded regions indicate the standard deviations from five independent runs.}
    \label{fig:ablation}
\end{figure}

{
\subsection{Learned Policy Function}
\label{appendix:policy_function}

To clarify how the agent responds to different labor market conditions, we examine the learned policy function. While the vacancy cost $c$ is a fixed structural parameter, the unemployment rate $u_t$ (where $u_t = 1 - l_t$) serves as the dynamic state variable conditioning the agent's decisions.

Figure~\ref{fig:policy_func} visualizes the learned policy function by plotting market tightness $\theta$ against the unemployment rate $u_t$. The resulting curve captures the agent's decision rule, exhibiting a clear counter-cyclical vacancy posting strategy where market tightness declines as unemployment rises, aligned with standard search-and-matching theory. Crucially, the learned policy curve is close to the theoretical steady-state equilibrium (indicated by the star marker). This alignment confirms that the calibrated agent has successfully recovered the optimal decision rule characterizing the competitive equilibrium.

\begin{figure}[ht]
    \centering
    \includegraphics[width=\columnwidth]{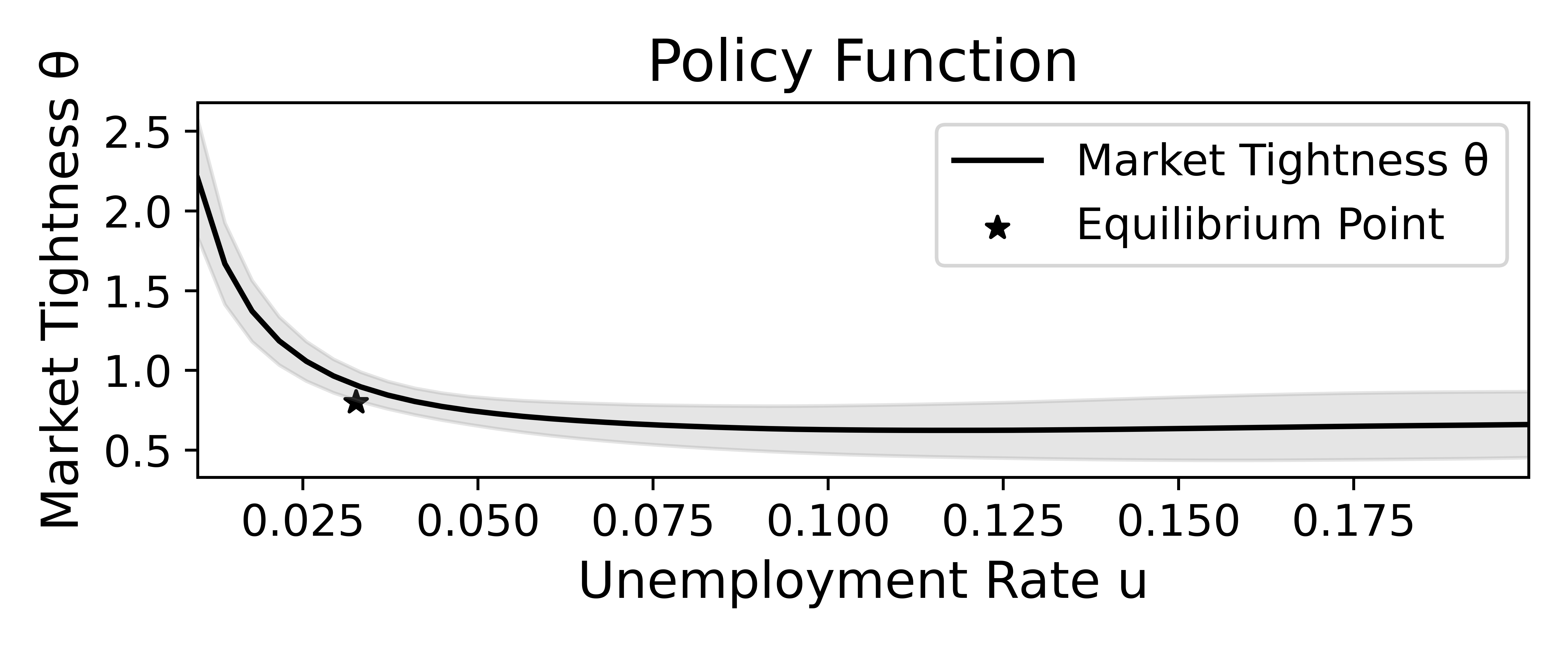}
    \caption{The learned policy function: Market Tightness $\theta$ vs. Unemployment Rate $u_t$. The  $\star$ marker indicates the theoretical steady-state equilibrium $(0.033, 0.78)$. The shaded regions indicate the standard deviations from five independent runs. }
    \label{fig:policy_func}
\end{figure}}

\section{Convergence of the RL-based Mean Field Fixed-Point Iteration}
\label{app:convergence}
\begin{definition}[Mean Field Update Operator]
Let $\Theta \subseteq \mathbb{R}^d$ denote the space of mean field parameters (e.g., aggregate employment rate, average wage). Let $\Pi$ denote the space of agent policies. Given a mean field $\theta \in \Theta$, define:
\begin{itemize}
    \item The optimal agent policy $\pi^*_\theta \in \Pi$ as the solution to a Markov Decision Process (MDP) with reward and transition dynamics parameterized by $\theta$.
    \item A deterministic mapping $\Phi: \Pi \to \Theta$ that returns the updated mean field resulting from population-level behavior under policy $\pi$.
\end{itemize}
Define the mean field update operator as the composite map:
\[
\Psi: \Theta \to \Theta, \quad \Psi(\theta) := \Phi(\pi^*_\theta)
\]
\end{definition}

\begin{assumption}[Lipschitz Continuity]
\label{ass:lipschitz}
Assume the following:
\begin{itemize}
    \item[(A1)] The policy mapping $\theta \mapsto \pi^*_\theta$ is Lipschitz continuous with constant $L_1 > 0$, i.e.,
    \[
    \|\pi^*_{\theta_1} - \pi^*_{\theta_2}\|_{\text{TV}} \leq L_1 \cdot \|\theta_1 - \theta_2\|
    \]
    for all $\theta_1, \theta_2 \in \Theta$.

    \item[(A2)] The mean field feedback $\Phi: \Pi \to \Theta$ is Lipschitz continuous with constant $L_2 > 0$, i.e.,
    \[
    \|\Phi(\pi_1) - \Phi(\pi_2)\| \leq L_2 \cdot \|\pi_1 - \pi_2\|_{\text{TV}}
    \]
    for all $\pi_1, \pi_2 \in \Pi$.
\end{itemize}
\end{assumption}

\begin{theorem}[Convergence to Mean Field Equilibrium]
\label{thm:convergence}
Suppose that Assumptions (A1)--(A2) hold, and that the composite Lipschitz constant $L := L_1 L_2 < 1$. Then:
\begin{enumerate}
    \item The mapping $\Psi: \Theta \to \Theta$ is a contraction mapping.
    \item There exists a unique fixed point $\theta^* \in \Theta$ such that $\Psi(\theta^*) = \theta^*$.
    \item For any initial value $\theta_0 \in \Theta$, the iterates $\theta_{k+1} := \Psi(\theta_k)$ converge exponentially fast to $\theta^*$:
    \[
    \|\theta_k - \theta^*\| \leq L^k \cdot \|\theta_0 - \theta^*\|
    \]
\end{enumerate}
\end{theorem}

\begin{proof}
From (A1) and (A2), we have for any $\theta_1, \theta_2 \in \Theta$:
\[
\begin{aligned}
\|\Psi(\theta_1) - \Psi(\theta_2)\|
&= \|\Phi(\pi^*_{\theta_1}) - \Phi(\pi^*_{\theta_2})\| \\
&\leq L_2 \cdot \|\pi^*_{\theta_1} - \pi^*_{\theta_2}\|_{\text{TV}} \\
&\leq L_2 L_1 \cdot \|\theta_1 - \theta_2\| = L \cdot \|\theta_1 - \theta_2\|
\end{aligned}
\]
Since $L < 1$, $\Psi$ is a contraction mapping on a complete metric space. By Banach’s Fixed Point Theorem, $\Psi$ admits a unique fixed point $\theta^*$ and the iteration $\theta_{k+1} = \Psi(\theta_k)$ converges to $\theta^*$ with exponential rate $L^k$. This completes the proof.
\end{proof}

\end{document}